\documentclass[12pt]{article}
\usepackage{bm}
\usepackage{amsmath}
\usepackage{amssymb, amscd, amsthm,amsfonts}
\usepackage[dvips]{graphicx}
\usepackage{verbatim}
\usepackage[dvips]{color}
\newtheorem{theorem}{Theorem}
\newtheorem{lemma}{Lemma}

\newtheorem{corollary}{Corollary}

\usepackage{pict2e}

\begin{document}

\title{On the Arikan Transformations of Binary-Input Discrete Memoryless Channels\footnote{This work was supported by the Natural Science Foundation of China (No. 61977056).}}
\author{Yadong Jiao, Xiaoyan Cheng, Yuansheng Tang\footnote{Corresponding author.}\\
{\it\small School of Mathematical Sciences, Yangzhou University, Jiangsu, China}\\
and Ming Xu\footnote{Email addresses: dx120210046@stu.yzu.edu.cn(Y. Jiao), xycheng@yzu.edu.cn(X. Cheng), ystang@yzu.edu.cn(Y. Tang), mxu@szcu.edu.cn(M. Xu)}
\\
{\it \small Suzhou City University, Jiangsu, China}}
\date{}
\maketitle

\begin{abstract}
The polar codes introduced by Arikan in 2009 achieve the capacity of binary-input discrete memoryless channels (BIDMCs) with low-complexity encoding and decoding. Identifying the unreliable synthetic channels, generated by Arikan transformation during the construction of these polar codes, is crucial. Currently, because of the large size of the output alphabets of synthetic channels, there is no efficient and practical approach to evaluate their reliability in general.
To tackle this problem, by converting the generation of synthetic channels in polar code construction into algebraic operations, in this paper we develop a method to characterize the synthetic channels
as random switching channels of binary symmetric channels when the underlying channels are symmetric.
Moreover, a lower bound for the average number of elements that possess the same likelihood ratio within the output alphabet of any synthetic channel generated in polar codes is also derived.
\end{abstract}
\vspace{1ex}
{\noindent\small{\bf Keywords:}
    	BIDMC; Polar code; Arikan transformation; likelihood ratio profile; random switching channel}

\section{Introduction}

Polar codes introduced by Arikan in 2009 are the first code family achieving the capacity of binary-input discrete memoryless channels (BIDMCs) under successive cancellation (SC) decoding as the code length tends to infinity. Due to their low complexity encoding and decoding, polar codes were selected in 2016 by the 3GPP Group for the uplink/downlink channel control in the 5G standard.

By using the kernel matrix $G$ of order 2, one can transform two independent BIDMCs into two synthetic channels (cf. \cite{Arikan09}) that are equivalents of the parity-constrained-input parallel channel and the parallel broadcast channel defined in \cite{Shamai05}. These synthetic channels, which and their compounds are called Arikan transformations of the underlying channels in this paper, preserve the sum of capacities.
The idea of polar codes proposed in \cite{Arikan09} is as follows. At first, $N=2^k$ independent
copies of the underlying BIDMC are iteratively transformed $k$ times into $N$ synthetic channels
by using the kernel matrix $G$. When $N$ is large enough, the synthetic channels
polarize into a set of reliable channels and a set of unreliable channels. Then, reliable communication is realized by transmitting the information bits only on the reliable channels while the remaining channels are {\it frozen} (i.e., their inputs are transparent for the receiver).

Clearly, for polar codes, it is critical to identify the unreliable synthetic channels.
In literature, several techniques have been proposed to estimate the reliability of the synthetic channels: Monte
Carlo simulation \cite{Arikan09}, density evolution and its Gaussian approximation
\cite{Tanaka09}, \cite{Trifonov12}, efficient degrading
and upgrading method \cite{Vardy13}, polarization weight and $\beta$-expansion \cite{3GPP16}, \cite{He17}, etc.
In general, the ranking of the synthetic channels depends
on the underlying channel. However, it was observed that
there is a partial order (with respect to degradation) between
the synthetic channels, which holds for any underlying
channel. By exploiting the partial order, the
complexity of the code construction can be significantly
reduced (cf. \cite{Urbanke19}, \cite{Siegel19}).

About the determination of the exact reliability for all the synthetic channels generated in polar codes, Arikan proposed in \cite{Arikan09} a recursive algorithm based on the Bhattacharyya parameter. However, this algorithm is efficient only for the case that the underlying channel is a binary erasure channel (BEC). As a simple corollary of the results in \cite{Blackwell1953}, it was shown in \cite{GR20} that two BIDMCs are mutually degraded if and only if their Blackwell measures are the same. For
the synthetic channels obtained in
one-step polarization, their Blackwell measures are investigated in \cite{GR20} by using functionals on the set of BIDMCs and determined entirely
when the underlying channels are symmetric.

In this paper, our main focus is on providing the compactest representations for the synthetic channels generated in polar codes over symmetric BIDMCs. Through these representations, the reliability of the synthetic channels can be computed efficiently.

The paper is organized as follows.
In Section~\ref{sec02}, we introduce a few definitions and some elementary properties of likelihood ratio profile, equivalence, degradation, random switching and symmetry of BIDMCs.
In Section~\ref{sec03}, we introduce some elementary properties of Arikan transformations of arbitrary BIDMCs and convert the Arikan transformations of symmetric BIDMCs into algebraic operations between binary symmetric channels (BSCs).
In Section~\ref{sec04}, we give a brief introduction for the polar codes from a new angle of view and a method for representing the Arikan transformations of symmetric BIDMCs in
the compactest form. A lower bound for the average number of elements that possess the same likelihood ratio within the output alphabet of any synthetic channel generated in polar codes is also derived in this section.
Finally, conclusions are drawn in Section \ref{sec05}.
\section{Binary-Input Discrete Memoryless Channels}
\label{sec02}
In this paper, for any random variable $v$ the notation $v$ may also express, a concrete value in its sample space, or the probability event that $v$ takes a concrete value, if there is no confusion.
For example, $\Pr(v)$ denotes the probability of that $v$ takes a concrete value which is also denoted by $v$.
In particular, if $u$ is a random variable whose sample space is a subset of that of $v$, then $v=u$ expresses the probability event that $v$ takes a concrete value denoted by $u$.

Let $W:x\in\mathcal{X}\mapsto y\in\mathcal{Y}$ be a \emph{binary-input discrete memoryless channel} (BIDMC), where the input $x$ is always supposed to be {\it uniformly distributed} in $\mathcal{X}=\mathbb{F}_2=\{0,1\}$, the finite field of two elements, and the output alphabet $\mathcal{Y}$ is a discrete set.

The {\it symmetric capacity} of $W$ is defined as
\begin{equation}\label{sym_cap} I(W)=\sum_{y\in\mathcal{Y}}\sum_{x\in\mathcal{X}}\frac{1}{2}\Pr(y|x)
\log_2\frac{\Pr(y|x)}{\frac{1}{2}\Pr(y|x=0)+\frac{1}{2}\Pr(y|x=1)},
\end{equation}
which is equal to the {\it mutual information}
$I(x;y)=H(x)+H(y)-H(x,y)$
between the output $y$ and the input $x$.

The \emph{maximum likelihood decoding} (MLD) of the BIDMC $W$ decodes $\hat{y}\in\mathcal{Y}$ into $d_{\text{mld}}(\mathcal{L}_W(\hat{y}))$, where $\mathcal{L}_W(\hat{y})=\Pr(y=\hat{y}|x=0)/\Pr(y=\hat{y}|x=1)$ is the \emph{likelihood ratio} of $\hat{y}$,
$d_{\text{mld}}(l)\in\mathcal{X}$ equals
0 if $l\geq 1$ and 1 otherwise.
Then, the {\it probability of error decoding} for MLD of $W$ is given by
\begin{equation}\label{error probability}
	P_{\epsilon}(W)=\frac{1}{2}\sum_{y\in\mathcal{Y}}\min\{\Pr(y|x=0),\Pr(y|x=1)\}.
\end{equation}

The reliability of $W$ can also be scaled by the Bhattacharyya parameter
\begin{gather}
Z(W)=\sum_{y\in\mathcal{Y}}\sqrt{\Pr(y|x=0)\Pr(y|x=1)}.
\label{Bp}
\end{gather}

\subsection{Likelihood Ratio Profile of BIDMC}
Since the input $x$ of $W$ is uniformly distributed in $\mathcal{X}$,
we have $\Pr(x=0)=\Pr(x=1)=1/2$ and $\Pr(y)=(\Pr(y|x=0)+\Pr(y|x=1))/2$.
Let
\begin{align}
S_W(y)=\frac{\Pr(x=0,y)}{\Pr(y)}=\frac{\Pr(y|x=0)}{\Pr(y|x=0)+\Pr(y|x=1)},
\end{align}
which is the posterior probability of $x=0$ given the output $y$ of $W$. Then, $S_W=S_W(y)$ is a random variable that takes values in the
unit interval $[0,1]$ and has mean
\begin{align}
\mathbf{E}(S_W)=\sum_{y\in\mathcal{Y}}S_W(y)\Pr(y)=\sum_{y\in\mathcal{Y}}\Pr(x=0,y)=1/2.
\end{align}
For $\varepsilon\in[0,1]$, let $P_W(\varepsilon)$ denote the probability $\Pr(S_W=\varepsilon)=\Pr(y\in L_W(\varepsilon))$,
where
\begin{align}
L_W(\varepsilon)&=\{y\in\mathcal{Y}: S_W(y)=\varepsilon\}\nonumber\\
&=\{y\in\mathcal{Y}: \mathcal{L}_W(y)=\varepsilon/\overline{\varepsilon}\}.
\end{align}
Since $\{L_W(\varepsilon)\}_{\varepsilon\in[0,1]}$ is a partition of the output alphabet $\mathcal{Y}$ distinguished by the likelihood ratios, the function $P_W(\varepsilon)$ over $[0,1]$ will be called the {\it likelihood ratio profile} (LRP) of $W$ (c.f. \cite{JCT24}, the original version of this paper).
For example, for $0\leq a<b\leq \overline{a}$, the BIDMC $W=\mathrm{B}_{a,b}$ depicted in Figure~1 has LRP
\begin{align*}
 P_W(\varepsilon)=\left\{
 \begin{array}{ll}
 \frac{1}{2}(a+b), & \text{if }\varepsilon=\frac{a}{a+b}, \\
 \frac{1}{2}(\overline{a}+\overline{b}), & \text{if }\varepsilon=\frac{\overline{a}}{\overline{a}+\overline{b}}, \\
 0,&\text{otherwise}.
 \end{array}
\right.
\end{align*}

\setlength{\unitlength}{0.3cm}
\begin{figure}[t]
\begin{center}
\begin{picture}(12,7)
\put(2,0){\vector(4,2){8}}
\put(2,4){\vector(4,-2){8}}
\put(2,0){\vector(4,0){8}}
\put(2,4){\vector(4,0){8}}

\put(1,3.8){\mbox{\tiny $0$}}
\put(1,-0.2){\mbox{\tiny $1$}}
\put(10.5,3.8){\mbox{\tiny $0$}}
\put(10.5,-0.2){\mbox{\tiny $1$}}

\put(5.5,4.4){\mbox{\tiny $a$}}
\put(5.5,-1){\mbox{\tiny $\overline{b}$}}
\put(3.2,2.5){\mbox{\tiny $\overline{a}$}}
\put(3.2,1.){\mbox{\tiny $b$}}

\end{picture}
\end{center}
\caption{Transition probabilities of $\mathrm{B}_{a,b}$.}
\end{figure}
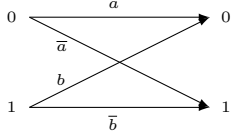

Note that
$\text{m}_{W}=\sum_{\varepsilon\in[0,1]}P_W(\varepsilon)\delta_{\varepsilon}$ is the  probability law of random variable $S_W$ and called the Blackwell measure of $W$ (cf. \cite{GR20}), where $\delta_{\varepsilon}$ denotes the Dirac measure centered on $\varepsilon$.

For any function $f:[0,1]\rightarrow \mathbb{R}$, let $I_f$ be the functional on the collection of BIDMCs defined by
\begin{align}
I_f(W)=\mathbf{E}(f(S_W))=\sum_{y\in\mathcal{Y}}f(S_W(y))\Pr(y)=\sum_{\varepsilon\in[0,1]}
f(\varepsilon)P_W(\varepsilon).
\end{align}
Then, many parameters for measuring of reliability of BIDMCs can be expressed as such functionals (c.f. \cite{GR20}).
For example, for any BIDMC $W$,
\begin{enumerate}
  \item $I_f(W)=I(W)$ if $f(\varepsilon)=1+\varepsilon\log_2 \varepsilon+\overline{\varepsilon}\log_2 \overline{\varepsilon}$;
  \item $I_f(W)=P_{\epsilon}(W)$ if $f(\varepsilon)=(1-|2\varepsilon-1|)/2=\min\{\varepsilon,\overline{\varepsilon}\}$;
  \item $I_f(W)=Z(W)$ if $f(\varepsilon)=2\sqrt{\varepsilon\overline{\varepsilon}}$.
\end{enumerate}

Since for any BIDMC the likelihood ratio constitutes a sufficient
statistic with respect to decoding (c.f.~\cite{Urbanke08}), two BIDMCs $W$, $W'$ will be said to be \emph{equivalent} as in \cite{JCT24}, and written as $W\cong W'$, if their LRPs are the same, i.e.,
$P_{W}(\varepsilon)=P_{W'}(\varepsilon)$ holds for all $\varepsilon\in[0,1]$.

For any BIDMC $W:x\in \mathcal{X}\mapsto y\in\mathcal{Y}$, another
BIDMC $W':x\in \mathcal{X}\mapsto y'\in\mathcal{Y}'$ is said a \emph{degradation channel} of $W$, and written as $W'\preccurlyeq W$,
if there is a channel $Q:\mathcal{Y}\rightarrow \mathcal{Y}'$ as depicted in Figure~2 such that
\begin{align}
\Pr(y'|x)=\sum_{y\in\mathcal{Y}}\Pr(y|x)Q(y'|y),\label{c05}
\end{align}
where $Q(b|a)$ is the probability of that $Q$ transmits $a\in\mathcal{Y}$ into $b\in\mathcal{Y}'$. The channel $Q$ is also called the \emph{intermediate channel} of the degradation $W'\preccurlyeq W$.
Clearly, from $W'\preccurlyeq W$ and $W\preccurlyeq W''$ one can get $W'\preccurlyeq W''$, i.e., the degradation of BIDMCs has {\it transitivity}. Furthermore, we have the following theorem (c.f. \cite{Blackwell1953}, \cite{GR20}, \cite{JCT24}).

\setlength{\unitlength}{0.3cm}
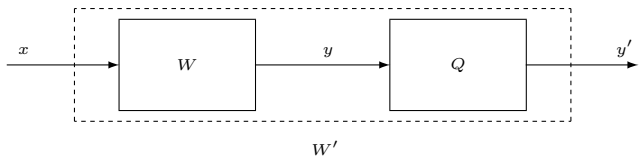
\begin{figure}[t]
\begin{center}
\begin{picture}(30,6.5)

\put(6,2){\framebox(6,4)[]{{\tiny $W$}}}
\put(18,2){\framebox(6,4)[]{{\tiny $Q$}}}
\put(14.5,0){{\tiny $W'$}}
\put(1.5,4.5){{\tiny $x$}}
\put(15,4.5){{\tiny $y$}}
\put(28,4.5){{\tiny $y'$}}
\put(1,4){\vector(1,0){5}}
\put(12,4){\vector(1,0){6}}
\put(24,4){\vector(1,0){5}}
\put(4,1.5){\dashbox{0.2}(22,5)[]}

\end{picture}
\end{center}
\caption{Degradation $W'\preccurlyeq W$ with intermediate channel $Q$.}
\end{figure}

\begin{theorem}\label{theo20}
For any BIDMCs $W:x\in\mathcal{X}\mapsto y\in\mathcal{Y}$ and $W':x\in\mathcal{X}\mapsto y'\in\mathcal{Y}'$,
we have $W\preccurlyeq W'\preccurlyeq W$ if and only if $W\cong W'$, i.e., $P_{W}(\varepsilon)$ equals $P_{W'}(\varepsilon)$ for any $\varepsilon\in[0,1]$.
\end{theorem}

\subsection{Random Switching of BIDMCs}
For any positive integer $n$, let $[n]$ denote the set $\{0,1,\cdots\,n-1\}$ of integers.

Assume that $\{\mathcal{Y}_j\}_{j\in[n]}$ is a partition of the output alphabet $\mathcal{Y}$ of the BIDMC $W$ such that, for each $j\in[n]$, the probability
\begin{equation*}
	\Pr(y\in\mathcal{Y}_j| x=a)=q_j
\end{equation*}
is independent of $a\in\mathcal{X}$. Clearly, we have $q_j=\Pr(y\in\mathcal{Y}_j),\, j\in[n]$ and that $(q_0,q_1,\ldots,q_{n-1})$ is a probability distribution vector. For $j\in[n]$, let $W_j:x\in\mathcal{X}\mapsto y_j\in\mathcal{Y}_j$ denote the synthetic BIDMC with transition probabilities
\begin{equation}\label{60}
	\Pr(y_j| x)=\frac{1}{q_j}\Pr(y=y_j| x),
\end{equation}
which, in accordance with our assumption regarding the notations, should be understood as
\begin{equation*}
	\Pr(y_j=b| x=a)=\frac{1}{q_j}\Pr(y=b| x=a), \text{ for any }b\in\mathcal{Y}_j\text{ and }a\in\mathcal{X}.\label{60'}
\end{equation*}

Since the input $x$ of the channel $W$ may be seen as randomly being transmitted over the sub-channels $\{W_j\}_{j\in[n]}$, and over $W_j$ with probability $q_j$ for each $j\in[n]$,
as depicted in Figure~3,
$W$ is also called a \emph{random switching channel} (RSC) of $\{W_j\}_{j\in[n]}$, and written as
\begin{align}
W=\sum_{j\in[n]}q_jW_j.\label{rr02}
\end{align}
We note that each pair of the output alphabets of the sub-channels $\{W_j\}_{j\in[n]}$ are disjoint, it is supposed naturally that the receiver knows the exact sub-channel used for each transmission.
Furthermore, the sub-channels $\{W_j\}_{j\in[n]}$ should be independent, i.e., their inputs
are independent random variables.

\setlength{\unitlength}{0.3cm}
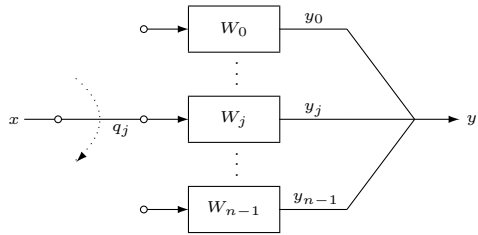
\begin{figure}[t]
\begin{center}
\begin{picture}(22,10)

\put(8,8){\framebox(4,2)[]{{\tiny $W_0$}}}
\put(8,4){\framebox(4,2)[]{{\tiny $W_j$}}}
\put(8,0){\framebox(4,2)[]{{\tiny $W_{n-1}$}}}
\put(10,2.5){\tiny $\vdots$}
\put(10,6.5){\tiny $\vdots$}
\multiput(6,1)(0,4){3}{\circle{.3}}
\multiput(6.15,1)(0,4){3}{\vector(1,0){1.85}}
\put(2.35,5){\line(1,0){3.5}}
\put(2.2,5){\circle{.3}}
\put(.7,5){\line(1,0){1.35}}
\qbezier[15](3.,3.2)(5.1,5)(3.,6.8)
\put(3.2,3.3){\vector(-3,-3){0.2}}
\put(0,4.8){\tiny $x$}
\multiput(12,1)(0,4){3}{\line(1,0){3}}
\put(18.,5){\line(-1,0){3}}
\put(18.,5){\line(-3,-4){3}}
\put(18.,5){\line(-3,4){3}}
\put(18.,5){\vector(1,0){2}}
\put(20.3,4.9){\tiny $y$}
\put(13.1,9.4){\tiny $y_0$}
\put(13.1,5.4){\tiny $y_j$}
\put(12.6,1.4){\tiny $y_{n-1}$}
\put(4.6,4.5){\tiny $q_j$}
\end{picture}
\end{center}
\caption{An RSC of BIDMCs $\{W_j\}_{j\in[n]}$. For each transmission, the sub-channel $W_j$ is chosen with probability $q_j$ over which the data is transmitted.}
\end{figure}

For the RSCs of BIDMCs, we have the following theorem
which shows the LRP and the notation defined by (\ref{rr02}) admit some natural operations (c.f. \cite{GR20}, \cite{JCT24}).

\begin{theorem}
Let $W=\sum_{j\in[n]}q_jW_j$ be an RSC of some independent BIDMCs $\{W_j\}_{j\in[n]}$.

\noindent
\begin{enumerate}
\item For any $\varepsilon\in[0,1]$, we have
\begin{align}
P_W(\varepsilon)=\sum_{j\in[n]}q_jP_{W_j}(\varepsilon).\label{62'}
\end{align}
\item Assume $W'=\sum_{j\in[n]}q'_jW'_j$ is a BIDMC that is independent of\ $W$, where $W'_j$ is equivalent to $W_j$ for each $j\in[n]$. Then,
for any $p\in[0,1]$ we have
\begin{align}
pW+\bar{p}W'\cong\sum_{j\in[n]}pq_jW_j+\sum_{j\in[n]}\overline{p}q'_jW'_j
\cong\sum_{j\in[n]}(pq_j+\overline{p}q'_j)W_j.\label{c00}
\end{align}
\end{enumerate}
\end{theorem}

For $\varepsilon\in[0,1]$, let $\mathrm{B}(\varepsilon)$ denote the BSC with crossover probability $ \varepsilon $ and $\mathrm{E}(\varepsilon)$ the BEC with erasure probability $\varepsilon$, respectively. Clearly, for any $\varepsilon,\sigma,q\in[0,1]$ we have $\mathrm{B}(\varepsilon)\cong\mathrm{B}(\overline{\varepsilon})\cong \mathrm{B}_{\varepsilon,\overline{\varepsilon}}$,
$q\mathrm{E}(\varepsilon)+\overline{q}\mathrm{E}(\sigma)\cong \mathrm{E}(q\varepsilon+\overline{q}\sigma)$ and
	$q\mathrm{B}(1/2)+\overline{q}\mathrm{B}(0)\cong \mathrm{E}(q)$, where $\mathrm{B}(0)$ and $\mathrm{B}(1/2)$ are the noiseless channel and the completely noisy channel respectively.

A BIDMC $W$ is said to be \emph{symmetric} (c.f. \cite{JCT24}) if its LRP is symmetric with respect to $1/2$, i.e.,
\begin{gather}
P_W(\varepsilon)=P_W(\overline{\varepsilon}),\ \text{for }\varepsilon\in[0,1].
\end{gather}
Clearly, the BIDMC $W$ is symmetric if and only if it is equivalent to an RSC of some BSCs.
For $n\geq 1$, let $\mathbb{B}_n$ denote the set of symmetric BIDMCs that are equivalents of RSCs of $n$ BSCs.
For $n>1$, let $\mathbb{B}^*_n=\mathbb{B}_n\setminus\mathbb{B}_{n-1}$.
Clearly, a BIDMC $W$ belongs to $\mathbb{B}^*_n$ if and only if
$W\cong\sum_{i\in[n]}p_i\mathrm{B}(\varepsilon_i)$ for some positive numbers $\{p_i\}_{i\in[n]}$ with sum 1 and $0\leq \varepsilon_1<\varepsilon_2<\cdots<\varepsilon_{n}\leq 1/2$.

Notice that the symmetry of BIDMCs defined here is based upon equivalence and slightly different from those defined in literature.
For example, for any $p\in(0,1)$ with $p\neq 1/2$ the channel given in Figure~4 is symmetric according to our definition, but according to that noted in \cite{Arikan09}.

\setlength{\unitlength}{0.3cm}
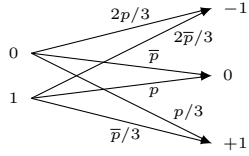
\begin{figure}[t]
\begin{center}
\begin{picture}(12,7)
\put(2,4){\vector(4,1){8}}
\put(2,4){\vector(8,-1){8}}
\put(2,4){\vector(4,-2){8}}
\put(2,2){\vector(4,2){8}}
\put(2,2){\vector(8,1){8}}
\put(2,2){\vector(4,-1){8}}
\put(1,3.8){\mbox{\tiny $0$}}
\put(1,1.8){\mbox{\tiny $1$}}
\put(10.5,5.8){\mbox{\tiny $-1$}}
\put(10.5,2.8){\mbox{\tiny $0$}}
\put(10.5,-0.2){\mbox{\tiny $+1$}}
\put(5.5,.2){\mbox{\tiny $\overline{p}/3$}}
\put(5.5,5.5){\mbox{\tiny $2p/3$}}
\put(8.3,1.2){\mbox{\tiny $p/3$}}
\put(8.3,4.5){\mbox{\tiny $2\overline{p}/3$}}
\put(7.2,2.2){\mbox{\tiny $p$}}
\put(7.2,3.7){\mbox{\tiny $\overline{p}$}}

\end{picture}
\end{center}
\caption{A symmetric BIDMC which is equivalent to $\mathrm{B}(p)$.}
\end{figure}

\section{Arikan Transformations of BIDMCs}
\label{sec03}

\subsection{Definition and LRPs of Arikan Transformations}
Let $(u_0,u_1)$ be a random vector uniformly distributed over $\mathcal{X}^2$.
Assume that $u_0+u_1$ is transmitted over $W_0:x_0\in\mathcal{X}\mapsto y_0\in\mathcal{Y}_0$, whereas $u_1$ is transmitted over $W_1:x_1\in\mathcal{X}\mapsto y_1\in\mathcal{Y}_1$, respectively,
where the BIDMCs $W_0$ and $W_1$ are independent.
Let $A_0(W_0,W_1)$ denote the synthetic BIDMC: $u_0\mapsto(y_0,y_1)$ with transition probabilities
\begin{align}
  \Pr(y_0,y_1|u_0)=\frac{1}{2}\sum_{u_1\in\mathcal{X}}\Pr(y_0|x_0=u_0+u_1)\Pr(y_1|x_1=u_1).\label{at06}
\end{align}
Let $A_1(W_0,W_1)$ denote the synthetic BIDMC: $u_1\mapsto(y_0,y_1,u_0)$ with transition probabilities
\begin{align}
  \Pr(y_0,y_1,u_0|u_1)=\frac{1}{2}\Pr(y_0|x_0=u_0+u_1)\Pr(y_1|x_1=u_1).\label{at07}
\end{align}

We note that $A_0(W_0,W_1)$ and $A_1(W_0,W_1)$ are just the synthetic BIDMCs denoted by $W_0\boxtimes W_1$ and $W_0\circledast W_1$, respectively, in \cite{Hirche18} and \cite{GR20}.
Since these synthetic BIDMCs played important roles in the proposal of Arikan's polar codes, we will call them and their compositions as \emph{Arikan Transformations} of some underlying BIDMCs as in \cite{JCT24}.
For example,
$A_0(A_1(W_0,W_1),W_2)$ is an Arikan transformation of independent BIDMCs $W_0,W_1,W_2$.
For the transformations $A_0(W_0,W_1)$ and $A_1(W_0,W_1)$, their Blackwell measures are investigated in \cite{GR20} via their functionals, we will give a direct method for computing their LRPs from those of $W_0$ and $W_1$.

For real numbers $a,b\in[0,1]$, let $a\star b=\bar{a}b+a\bar{b}$.
Clearly, we have $a\star b=b\star a$
and
\begin{gather*}
\overline{a}\star b=ab+\overline{a}\overline{b}=
(a+\overline{a})(b+\overline{b})-(\bar{a}b+a\bar{b})
=\overline{a\star b},\\
(a\star b)\star c=\bar{a}bc+a\bar{b}c+ab\overline{c}+\overline{a}\overline{b}\overline{c}
=a\star(b\star c).
\end{gather*}

\begin{theorem}\label{theo50}
Let $A_0=A_0(W_0,W_1)$ and $A_1=A_1(W_0,W_1)$ be the Arikan transformations of independent BIDMCs $W_0,W_1$. For any $\varepsilon\in[0,1]$, we have
\begin{align}
P_{A_0}(\varepsilon)=&\sum_{\overline{\varepsilon_0\star\varepsilon_1}=\varepsilon}
P_{W_0}(\varepsilon_0)P_{W_1}(\varepsilon_1),\label{ggg1}\\
P_{A_1}(\varepsilon)=&\sum_{\varepsilon_0\varepsilon_1/
\overline{\varepsilon_0\star\varepsilon_1}=\varepsilon}
\overline{\varepsilon_0\star\varepsilon_1}
P_{W_0}(\varepsilon_0)P_{W_1}(\varepsilon_1)\nonumber\\
&+\sum_{\overline{\varepsilon_0}\varepsilon_1/
(\varepsilon_0\star\varepsilon_1)=\varepsilon}
(\varepsilon_0\star\varepsilon_1)
P_{W_0}(\varepsilon_0)P_{W_1}(\varepsilon_1).\label{ggg2}
\end{align}
\end{theorem}
\begin{proof}
Clearly, for $y_0\in L_{W_0}(\varepsilon_0)$ and $y_1\in L_{W_1}(\varepsilon_1)$, we have
\begin{align}
&\frac{\Pr(y_0,y_1|u_0=0)}{\Pr(y_0,y_1|u_0=1)}
=\frac{\Pr(y_0|x_0=0)\Pr(y_1|x_1=0)+\Pr(y_0|x_0=1)\Pr(y_1|x_1=1)}
{\Pr(y_0|x_0=1)\Pr(y_1|x_1=0)+\Pr(y_0|x_0=0)\Pr(y_1|x_1=1)}\nonumber\\
&=\frac{\varepsilon_0/\overline{\varepsilon_0}\cdot \varepsilon_1/\overline{\varepsilon_1}+1}
{\varepsilon_0/\overline{\varepsilon_0}+\varepsilon_1/\overline{\varepsilon_1}}
=\frac{\,\overline{\varepsilon_0\star\varepsilon_1}\,}
{\varepsilon_0\star\varepsilon_1},
\label{ggg3}\\
&\frac{\Pr(y_0,y_1,u_0=0|u_1=0)}{\Pr(y_0,y_1,u_0=0|u_1=1)}
=\frac{\Pr(y_0|x_0=0)\Pr(y_1|x_1=0)}{\Pr(y_0|x_0=1)\Pr(y_1|x_1=1)}
=\frac{\varepsilon_0\varepsilon_1}{\overline{\varepsilon_0}\,\overline{\varepsilon_1}},\\
&\frac{\Pr(y_0,y_1,u_0=1|u_1=0)}{\Pr(y_0,y_1,u_0=1|u_1=1)}
=\frac{\Pr(y_0|x_0=1)\Pr(y_1|x_1=0)}{\Pr(y_0|x_0=0)\Pr(y_1|x_1=1)}
=\frac{\overline{\varepsilon_0}\varepsilon_1}{\varepsilon_0\overline{\varepsilon_1}},
\end{align}
and
\begin{align}
&\Pr(y_0,y_1)=(\Pr(y_0,y_1|u_0=0)+\Pr(y_0,y_1|u_0=1))/2\nonumber\\
&=(\Pr(y_0|x_0=0)+\Pr(y_0|x_0=1))(\Pr(y_1|x_1=0)+\Pr(y_1|x_1=1))/4\nonumber\\
&=\Pr(y_0)\Pr(y_1),\\
&\Pr(y_0,y_1,u_0=0)=(\Pr(y_0,y_1,u_0=0|u_1=0)+\Pr(y_0,y_1,u_0=0|u_1=1))/2\nonumber\\
&=(\Pr(y_0|x_0=0)\Pr(y_1|x_1=0)+\Pr(y_0|x_0=1)\Pr(y_1|x_1=1))/4\nonumber\\
&=\overline{\varepsilon_0\star\varepsilon_1}\Pr(y_0)\Pr(y_1),\\
&\Pr(y_0,y_1,u_0=1)=(\Pr(y_0,y_1,u_0=1|u_1=0)+\Pr(y_0,y_1,u_0=1|u_1=1))/2\nonumber\\
&=(\Pr(y_0|x_0=1)\Pr(y_1|x_1=0)+\Pr(y_0|x_0=0)\Pr(y_1|x_1=1))/4\nonumber\\
&=(\varepsilon_0\star\varepsilon_1)\Pr(y_0)\Pr(y_1).
\label{ggg4}
\end{align}
Therefore, according to
\begin{gather*}
\sum_{y_0\in L_{W_0}(\varepsilon_0),\, y_1\in L_{W_1}(\varepsilon_1)}\Pr(y_0)\Pr(y_1)
=P_{W_0}(\varepsilon_0)P_{W_1}(\varepsilon_1),
\end{gather*}
we see that (\ref{ggg1}) and (\ref{ggg2}) follow from (\ref{ggg3}) to (\ref{ggg4}).
\end{proof}

The following theorem is a simple corollary of Theorem~\ref{theo50}.
\begin{theorem}\label{lem200}
If $W_0=\sum_{i\in[n]}p_iW'_i$ and $W_1=\sum_{j\in[m]}q_jW''_j$, where $W'_i$ and $W''_j$ are independent BIDMCs for any $i\in[n]$ and $j\in[m]$, then
\begin{align}\label{at00}
A_0(W_0,W_1)\cong\sum_{i\in[n],j\in[m]}p_iq_jA_0(W'_i,W''_j),\\
A_1(W_0,W_1)\cong\sum_{i\in[n],j\in[m]}p_iq_jA_1(W'_i,W''_j).\label{at01}
\end{align}
\end{theorem}

From $a\star b=b\star a$, $a\star \frac{1}{2}=\frac{1}{2}$, $\overline{\overline{a\star b}\star c}=\overline{a\star \overline{b\star c}}$ and (\ref{ggg1}) we have
\begin{gather}
A_0(W_0,W_1)\cong A_0(W_1,W_0),\label{mm303}\\
  A_0(\mathrm{B}(1/2),W)\cong \mathrm{B}(1/2),\label{mm300} \\
A_0(A_0(W_0,W_1),W_2)\cong A_0(W_0,A_0(W_1,W_2)).\label{mm305}
\end{gather}
However, $A_1(W_0,W_1)$ may be not equivalent to $A_1(W_1,W_0)$.
For example, for $0<a<b<\overline{a}$ and $$\frac{\sqrt{a\overline{b}}}{\sqrt{\overline{a}b}+\sqrt{a\overline{b}}}<\sigma<
\frac{\overline{b}}{\overline{a}+\overline{b}},$$
let $W=A_1(\mathrm{B}(\sigma),\mathrm{B}_{a,b})$ and $W'=A_1(\mathrm{B}_{a,b},\mathrm{B}(\sigma))$, then we have
\begin{gather*}
P_{W}(\varepsilon)=\left\{
\begin{array}{ll}
\frac{a\sigma+b\overline{\sigma}}{2},&\text{if }
\varepsilon=\frac{a\sigma}{a\sigma+b\overline{\sigma}},\\
\frac{a\overline{\sigma}+b\sigma}{2},&\text{if }
\varepsilon=\frac{a\overline{\sigma}}{a\overline{\sigma}+b\sigma},\\
\frac{\overline{a}\sigma+\overline{b}\overline{\sigma}}{2},&\text{if }
\varepsilon=\frac{\overline{a}\sigma}{\overline{a}\sigma+\overline{b}\overline{\sigma}},\\
\frac{\overline{a}\,\overline{\sigma}+\overline{b}\sigma}{2},&\text{if }
\varepsilon=\frac{\overline{a}\,\overline{\sigma}}
{\overline{a}\,\overline{\sigma}+\overline{b}\sigma},\\
0,&\text{otherwise},
\end{array}
\right.\ \
P_{W'}(\varepsilon)=\left\{
\begin{array}{ll}
\frac{a\sigma+b\overline{\sigma}}{4},&\text{if }
\varepsilon=\frac{a\sigma}{a\sigma+b\overline{\sigma}},\\
\frac{a\sigma+b\overline{\sigma}}{4},&\text{if }
\varepsilon=\frac{b\overline{\sigma}}{a\sigma+b\overline{\sigma}},
\\
\frac{a\overline{\sigma}+b\sigma}{4},&\text{if }
\varepsilon=\frac{a\overline{\sigma}}{a\overline{\sigma}+b\sigma},\\
\frac{a\overline{\sigma}+b\sigma}{4},&\text{if }
\varepsilon=\frac{b\sigma}{a\overline{\sigma}+b\sigma},
\\
\frac{\overline{a}\sigma+\overline{b}\overline{\sigma}}{4},&\text{if }
\varepsilon=\frac{\overline{a}\sigma}{\overline{a}\sigma+\overline{b}\overline{\sigma}},\\
\frac{\overline{a}\sigma+\overline{b}\overline{\sigma}}{4},&\text{if }
\varepsilon=\frac{\overline{b}\overline{\sigma}}{\overline{a}\sigma+\overline{b}\overline{\sigma}},
\\
\frac{\overline{a}\,\overline{\sigma}+\overline{b}\sigma}{4},&\text{if }
\varepsilon=\frac{\overline{a}\,\overline{\sigma}}
{\overline{a}\,\overline{\sigma}+\overline{b}\sigma},\\
\frac{\overline{a}\,\overline{\sigma}+\overline{b}\sigma}{4},&\text{if }
\varepsilon=\frac{\overline{b}\sigma}
{\overline{a}\,\overline{\sigma}+\overline{b}\sigma},
\\
0,&\text{otherwise}.
\end{array}
\right.
\end{gather*}
Hence $W'$ is not equivalent to $W$. Indeed, $W'$
is symmetric while $W$ is asymmetric.

\subsection{Arikan Transformations of Symmetric BIDMCs}

In this subsection we deal with the Arikan transformations when the underlying channels are symmetric BIDMCs.
For real numbers $a,b\in[0,1]$, let
\begin{gather}a\diamond b=\left\{\begin{array}{ll}
ab/(\overline{a}\star b),&\text{if }\{a,b\}\subset(0,1),\\
0,&\text{if }\{a,b\}\cap \{0,1\}\neq \emptyset.
\end{array}\right.\end{gather}
One can deduce the following lemma easily.
\begin{lemma}\label{lem201}
For the operation $\diamond$, we have
\begin{gather}
a\diamond b=b\diamond a,\ (a\diamond b)\diamond c=a \diamond(b\diamond c),\label{gg1}\\
 \overline{a\diamond b}=\overline{a}\diamond \overline{b}\text{ for }\{a,b\}\subset (0,1),\label{gg3}\\
a\diamond \overline{a}=\frac{1}{2}\text{ and }a\diamond \frac{1}{2}=a\ \text{for }a\in(0,1).\label{gg4}
\end{gather}
Furthermore, $\langle(0,1),\diamond\rangle$ is an Abel group and
$$
\langle(0,1),\star\rangle,\ \langle[0,1/2],\star\rangle,\ \langle[0,1/2],\diamond\rangle
$$
are commutative semi-groups.
\end{lemma}

When the underlying BIDMCs are BSCs, from Theorem~\ref{theo50} one can deduce easily that the Arikan transformations assume simple forms as presented in the following lemma (c.f. \cite{GR20}, \cite{JCT24}).
\begin{lemma}\label{lem202}
For any $\varepsilon_0,\varepsilon_1\in[0,1]$, we have
\begin{gather}
A_0(\mathrm{B}(\varepsilon_0),\mathrm{B}(\varepsilon_1))
\cong\mathrm{B}({\varepsilon_0}\star\varepsilon_1)\label{at31},\\
A_1(\mathrm{B}(\varepsilon_0),\mathrm{B}(\varepsilon_1))\cong
({\varepsilon_0}\star\overline{\varepsilon_1})
\mathrm{B}({\varepsilon_0}\diamond{\varepsilon_1})
+({\varepsilon_0}\star{\varepsilon_1})
\mathrm{B}({\varepsilon_0}\diamond\overline{\varepsilon_1}).\label{at32}
\end{gather}
\end{lemma}

Furthermore, for the Arikan transformations whose underlying BIDMCs are symmetric, we have the following lemma.
\begin{lemma}\label{lem301}
For symmetric BIDMCs $W$, $W'$ and $W''$, we have
\begin{gather}
  A_1(\mathrm{B}(1/2),W)\cong
  A_0(\mathrm{B}(0),W)\cong W, \label{mm301}\\
  A_1(\mathrm{B}(0),W)\cong \mathrm{B}(0),\label{mm302}\\
A_1(W,W')\cong A_1(W',W),\label{mm304}\\
A_1(A_1(W,W'),W'')\cong A_1(W,A_1(W',W'')).\label{mm306}
\end{gather}
In particular, the BIDMCs $A_0(W,W')$ and $A_1(W,W')$ are symmetric.
\end{lemma}
\begin{proof}
From Lemmas~\ref{lem201} and \ref{lem202}, we see easily that (\ref{mm301}) to (\ref{mm304}) are valid when $W,W',W''$ are BSCs.
According to (\ref{mm302}) we see that (\ref{mm306}) is valid when at least one channel in $\{W,W',W''\}$
is $\mathrm{B}(0)$.
Furthermore, for $\varepsilon,\sigma,\delta\in(0,1)$, from $({\varepsilon}\star\overline{\sigma})({\varepsilon}\diamond{\sigma})
={\varepsilon}{\sigma}$, ${\varepsilon}\diamond\overline{\sigma}\diamond\overline{\delta}=
\overline{\overline{\varepsilon}\diamond{\sigma}\diamond\delta}$, (\ref{at01}) and (\ref{at32}) we have
\begin{align*}
&A_1(A_1(\mathrm{B}(\varepsilon),\mathrm{B}(\sigma)),\mathrm{B}(\delta))\\
\cong&A_1\big(({\varepsilon}\star\overline{\sigma})
\mathrm{B}({\varepsilon}\diamond{\sigma})
+({\varepsilon}\star{\sigma})
\mathrm{B}({\varepsilon}\diamond\overline{\sigma}),\mathrm{B}(\delta)\big)\\
\cong&({\varepsilon}\star\overline{\sigma})\big((({\varepsilon}\diamond{\sigma})\star\overline{\delta})
\mathrm{B}(({\varepsilon}\diamond{\sigma})\diamond{\delta})
+(({\varepsilon}\diamond{\sigma})\star{\delta})
\mathrm{B}((\varepsilon\diamond{\sigma})\diamond\overline{\delta})\big)+\\
&\ \ \ ({\varepsilon}\star{\sigma})\big((({{\varepsilon}}\diamond\overline{\sigma})\star\overline{\delta})
\mathrm{B}(({{\varepsilon}}\diamond\overline{\sigma})\diamond{\delta})
+(({{\varepsilon}}\diamond\overline{\sigma})\star{\delta})
\mathrm{B}(({\varepsilon}\diamond\overline{\sigma})\diamond\overline{\delta})\big)\\
\cong&(\varepsilon\sigma\delta+\overline{\varepsilon}\,\overline{\sigma}\overline{\delta})
\mathrm{B}(\varepsilon\diamond\sigma\diamond\delta)+
(\varepsilon\sigma\overline{\delta}+\overline{\varepsilon}\,\overline{\sigma}{\delta})
\mathrm{B}(\varepsilon\diamond\sigma\diamond\overline{\delta})+\\
&\ \ \ ({\varepsilon}\overline{\sigma}\delta+\overline{\varepsilon}{\sigma}\overline{\delta})
\mathrm{B}({\varepsilon}\diamond\overline{\sigma}\diamond\delta)+
(\varepsilon\overline{\sigma}\overline{\delta}+\overline{\varepsilon}{\sigma}{\delta})
\mathrm{B}(\overline{\varepsilon}\diamond{\sigma}\diamond\delta)
\end{align*}
and thus we see that (\ref{mm306}) is valid when $W,W',W''$ are BSCs.
Hence, from Theorem \ref{lem200} we see that (\ref{mm301}) to (\ref{mm306}) are valid for any symmetric BIDMCs $W,W',W''$.
Clearly, the BIDMCs $A_0(W,W')$ and $A_1(W,W')$ are symmetric since they are equivalent to some RSCs of BSCs.
\end{proof}

For independent and symmetric BIDMCs $W_0,W_1,\ldots$, we define further $\Delta(W_0)=\nabla(W_0)=W_0$
and, for $t\geq 1$,
\begin{gather}
\Delta(\{W_i\}_{i\in[t+1]})=A_0(\Delta(\{W_i\}_{i\in[t]}),W_{t}),\label{u01}\\
\nabla(\{W_i\}_{i\in[t+1]})=A_1(\nabla(\{W_i\}_{i\in[t]}),W_{t}).\label{u02}
\end{gather}
Clearly,
for $1\leq j<t$ we have
\begin{gather}
\Delta(\{W_i\}_{i\in[t]})\cong A_0(\Delta(\{W_i\}_{i\in[j]}),\Delta(\{W_{i+j}\}_{i\in[t-j]})),\label{u05}\\
\nabla(\{W_i\}_{i\in[t]})\cong A_1(\nabla(\{W_i\}_{i\in[j]}),\nabla(\{W_{i+j}\}_{i\in[t-j]}),\label{u06}
\end{gather}
and for any permutation $\{W'_i\}_{i\in[t]}$ of the channels $\{W_i\}_{i\in[t]}$ we have
\begin{gather}
\Delta(\{W'_i\}_{i\in[t]})\cong \Delta(\{W_i\}_{i\in[t]}),\label{u03}\\
\nabla(\{W'_i\}_{i\in[t]})\cong \nabla(\{W_i\}_{i\in[t]}).\label{u04}
\end{gather}

If no confusion, the vector $(x_0,x_1,\ldots,x_{n-1})$ is also denoted by $x_0^n$.
For any vector $\sigma_0^t\in[0,1]^t$, let
\begin{align}
\varpi(\sigma_0^t)=\left\{
\begin{array}{ll}
\frac{1}{2}\big(\prod_{i\in [t]}\sigma_i+\prod_{i\in[t]}\overline{\sigma_i}\big),
& \text{if }\sigma_0^t\in(0,1)^t,
\\
2^{-t},
& \text{if }\sigma_0^t\not\in(0,1)^t.
\end{array}\right.\label{u07}
\end{align}
For $\sigma_0^{t+1}\in(0,1)^{t+1}$, from the definition of the operation $\diamond$, by induction one can prove easily
\begin{gather*}
\diamond_{i\in[t]}\sigma_{i}=\frac{\prod_{i\in[t]}\sigma_{i}}
{\prod_{i\in[t]}\sigma_{i}+\prod_{i\in[t]}\overline{\sigma_{i}}}
=\frac{\prod_{i\in[t]}\sigma_{i}}{2\varpi(\sigma_0^t)},
\end{gather*}
and then from $\overline{\diamond_{i\in[t]}\sigma_{i}}=\diamond_{i\in[t]}\overline{\sigma_{i}}$ we see
\begin{gather}
\varpi(\sigma_0^t)\big((\diamond_{i\in[t]}\sigma_{i})\star\overline{\sigma_t}\big)=
\varpi(\sigma_0^t)\bigg(\frac{\prod_{i\in[t+1]}\sigma_{i}}{2\varpi(\sigma_0^t)}
+\frac{\prod_{i\in[t+1]}\overline{\sigma_{i}}}{2\varpi(\sigma_0^t)}\bigg)
=\varpi(\sigma_0^{t+1}).\label{u10}
\end{gather}
\begin{lemma}
\label{lem700}
For $t\geq 1$ and $\varepsilon_0^t\in[0,1]^t$,
\begin{gather}
\Delta(\{\mathrm{B}(\varepsilon_i)\}_{i\in[t]})
\cong\mathrm{B}(\star_{i\in[t]}\varepsilon_i),\label{u08}\\
\nabla(\{\mathrm{B}(\varepsilon_i)\}_{i\in[t]})
\cong\sum_{\sigma_i\in\{\varepsilon_i,\overline{\varepsilon_i}\},i\in[t]}\varpi(\sigma_0^t)
\mathrm{B}(\diamond_{i\in[t]}\sigma_i),\label{u09}
\end{gather}
where $\varepsilon$ and $\overline{\varepsilon}$ should be seen as different elements even if $\varepsilon=1/2$.
\end{lemma}
\begin{proof}
Clearly, (\ref{u08}) follows from (\ref{at31}) and (\ref{u01}).
It is obvious that (\ref{u09}) is valid for $t=1$ or $\varepsilon_0^t\not\in(0,1)^t$.
Now we manage to prove (\ref{u09}) for $\varepsilon_0^t\in(0,1)^t$ by induction on $t$.
Indeed, according to (\ref{at01}), (\ref{at32}), (\ref{u02}) and (\ref{u10}) we have
\begin{align*}
&\nabla(\{\mathrm{B}(\varepsilon_i)\}_{i\in[t+1]})
=A_1(\nabla(\{\mathrm{B}(\varepsilon_i)\}_{i\in[t]}),\mathrm{B}(\varepsilon_t))\\
\cong&\sum_{\sigma_i\in\{\varepsilon_i,\overline{\varepsilon_i}\},i\in[t]}\varpi(\sigma_0^t)
A_1(\mathrm{B}(\diamond_{i\in[t]}\sigma_i),\mathrm{B}(\varepsilon_t))\\
\cong&\sum_{\sigma_i\in\{\varepsilon_i,\overline{\varepsilon_i}\},i\in[t]}\varpi(\sigma_0^t)
\sum_{\sigma_t\in\{\varepsilon_t,\overline{\varepsilon_t}\}}
((\diamond_{i\in[t]}\sigma_i)\star\overline{\sigma_t})\mathrm{B}(\diamond_{i\in[t+1]}\sigma_i)\\
\cong&\sum_{\sigma_i\in\{\varepsilon_i,\overline{\varepsilon_i}\},i\in[t+1]}\varpi(\sigma_0^{t+1})
\mathrm{B}(\diamond_{i\in[t+1]}\sigma_i).
\end{align*}

The proof is complete.
\end{proof}

The following theorem is a simple corollary of Theorem~\ref{lem200} and Lemma~\ref{lem700}.

\begin{theorem}\label{cor501'}
Let $\{W_l\}_{l\in [m]}$ be independent and symmetric BIDMCs.
Suppose $W_l\cong\sum_{i\in[n_l]}q_{i,l}\mathrm{B}(\varepsilon_{i,l})$ for $l\in[m]$, then we have
\begin{gather}
\Delta(\{W_l\}_{l\in[m]})\cong
\sum_{\{i_l\in[n_l]\}_{l\in[m]}}\bigg(\prod_{l\in[m]}q_{i_l,l}\bigg)
\mathrm{B}(\star_{l\in[m]}\varepsilon_{i_l,l}),
\label{mm01}\\
\nabla(\{W_l\}_{l\in[m]})\cong\sum_{\{i_l\in[n_l],\sigma_{l}\in
\{\varepsilon_{i_l},\overline{\varepsilon_{i_l}}\}\}_{l\in[m]}}
\bigg(\prod_{l\in [m]}q_{i_l,l}\bigg)\varpi(\sigma_0^m)\mathrm{B}(\diamond_{l\in[m]}\sigma_{l}),\label{mm02}
\end{gather}
where $\varepsilon$ and $\overline{\varepsilon}$ should be seen as different elements even if $\varepsilon=1/2$.
\end{theorem}

Notice that the synthetic BIDMCs $\Delta(\{W_l\}_{l\in[m]})$ and $\nabla(\{W_l\}_{l\in[m]})$ are equivalent to the parity-constrained-input parallel channel and the parallel broadcast channel discussed in \cite{Shamai05}, respectively, when the underlying channels $\{W_l\}_{l\in[m]}$ are symmetric BIDMCs.

\section{Synthetic Channels in Polar Codes}\label{sec04}

In this section, we mainly consider to express the synthetic channels, which are generated in the polar codes proposed by Arikan in
\cite{Arikan09}, as RSCs of only a few BSCs when the underlying channels are symmetric BIDMCs.

\subsection{Construction of Polar Codes}
At first, we give a brief introduction for the polar codes from a new angle of view.
The synthetic channels in
the polar code of order $k\geq 2$ are Arikan transformations generated iteratively from underlying channels $\{W^{(0)}_{\alpha}\}_{\alpha\in \{0,1\}^k}$ which are independent copies of a given BIDMC $W$.
For any $0\leq m\leq k-2$, $a,b\in \{0,1\}$, $\beta\in \{0,1\}^{m}$ and $\alpha\in \{0,1\}^{k-m-2}$, let
\begin{align}
W^{(m+1)}_{\alpha\beta a b}=A_a(W^{(m)}_{\alpha b\beta0},W^{(m)}_{\alpha b\beta1}).
\end{align}
For any $\beta\in \{0,1\}^{k-1}$ and $a\in \{0,1\}$, let
\begin{align}
W^{(k)}_{\beta a}=A_a(W^{(k-1)}_{\beta0},W^{(k-1)}_{\beta1}).
\end{align}
Notice that, for any $0\leq m\leq k-1$ and $\beta\in\{0,1\}^m$, the channels in
$$\{W^{(m)}_{\alpha\beta a}:\alpha\in\{0,1\}^{k-1-m},a\in\{0,1\}\}$$ are independent and equivalent to each other.

If the tuples $\alpha=(d_{0},\ldots,d_{k-1})\in \{0,1\}^k$ are ordered according to the ascending order of $$b(\alpha)=\sum_{i\in [k]}d_i2^{k-1-i},$$
the synthetic channels generated in polar code of order $k=4$ are shown in the Figure~5, where $A_4$ and $A_8$ denote transformations with structures similar to that shown in the left sides respectively.

\setlength{\unitlength}{0.3cm}
\newcounter{dc0}
\newcounter{dc4}
\begin{figure}[t]
\begin{center}
\ \ \ \ \ \ \begin{picture}(33,24)
\addtocounter{dc0}{-1}
\multiput(-1,21.9)(2,0){16}
{\addtocounter{dc0}{1}\makebox(0,0)[bl]{{\tiny $W^{(0)}_{\arabic{dc0}}$}}}
\addtocounter{dc4}{-1}
\multiput(-1,0.4)(2,0){16}
{\addtocounter{dc4}{1}\makebox(0,0)[bl]{{\tiny $W^{(4)}_{\arabic{dc4}}$}}}

\multiput(0,23.5)(2,0){16}{\vector(0,-1){2.5}}

\multiput(0,20)(4,0){2}{\framebox(2,1)[]{{\tiny $A$}}}
\multiput(0,16)(4,0){2}{\framebox(2,1)[]{{\tiny $A$}}}
\put(0,20){\vector(0,-1){3}}
\put(2,20){\vector(2,-3){2}}
\put(4,20){\vector(-2,-3){2}}
\put(6,20){\vector(0,-1){3}}

\multiput(0,10)(4,0){4}{\framebox(2,1)[]{{\tiny $A$}}}
\put(0,16){\vector(0,-1){5}}
\put(2,16){\vector(2,-5){2}}
\put(4,16){\vector(4,-5){4}}
\put(6,16){\vector(6,-5){6}}
\put(8,16){\vector(-6,-5){6}}
\put(10,16){\vector(-4,-5){4}}
\put(12,16){\vector(-2,-5){2}}
\put(14,16){\vector(0,-1){5}}
\put(8,16){\framebox(6,5){$A_4$}}

\multiput(0,2)(4,0){8}{\framebox(2,1)[]{{\tiny $A$}}}
\put(0,10){\vector(0,-1){7}}
\put(2,10){\vector(2,-7){2}}
\put(4,10){\vector(4,-7){4}}
\put(6,10){\vector(6,-7){6}}
\put(8,10){\vector(8,-7){8}}
\put(10,10){\vector(10,-7){10}}
\put(12,10){\vector(12,-7){12}}
\put(14,10){\vector(14,-7){14}}
\put(16,10){\vector(-14,-7){14}}
\put(18,10){\vector(-12,-7){12}}
\put(20,10){\vector(-10,-7){10}}
\put(22,10){\vector(-8,-7){8}}
\put(24,10){\vector(-6,-7){6}}
\put(26,10){\vector(-4,-7){4}}
\put(28,10){\vector(-2,-7){2}}
\put(30,10){\vector(0,-7){7}}
\put(16,10){\framebox(14,11){$A_8$}}
\multiput(0,2)(2,0){16}{\vector(0,-1){2.5}}

\end{picture}\end{center}
\caption{Synthetic channels in polar code of order $k=4$.}
\end{figure}
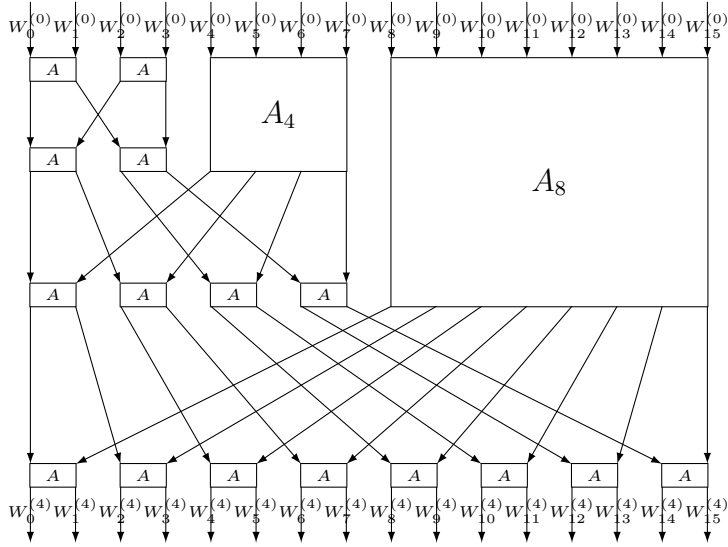

For $W_0\cong W_1\cong W$, we also write $A_0(W_0,W_1)$ and $A_1(W_0,W_1)$ as $A_0(W)$ and $A_1(W)$, respectively. For $n\geq 1$ and $\alpha\in\{0,1\}^n$, we define further $A_{\alpha 0}(W)=A_0(A_{\alpha}(W))$ and $A_{\alpha 1}(W)=A_1(A_{\alpha}(W))$ iteratively.
The following lemma is from \cite{Arikan09}.
\begin{lemma}
If the underlying channels $\{W^{(0)}_{\alpha}\}_{\alpha\in \{0,1\}^k}$ are copies of a given BIDMC $W$, then the set $\{W^{(k)}_{\alpha}:\alpha\in \{0,1\}^k\}$ of the synthetic channels is just the set $\{A_{\alpha }(W):\alpha\in \{0,1\}^k\}$
which shall exhibit polarization phenomena when $k$ is large: The capacity of the part with a proportion close to $I(W)$ is approximately equal to 1, and meanwhile the capacity of the part with a proportion close to $1-I(W)$ is approximately equal to 0.
\end{lemma}

For $d\in \{0,1\}$, let $\mathcal{S}(d)$ denote $\{1\}$ if $d=1$ and $\{0,1\}$ if $d=0$.
For $\delta=(d_0,\ldots,d_{n-1})\in \{0,1\}^n$, let
\begin{gather*}
\mathcal{S}(\delta)=\{(a_{0},\ldots,a_{n-1})\in \{0,1\}^n: a_i\in \mathcal{S}(d_{n-1-i}), i\in [n]\}.
\end{gather*}
For example, $\mathcal{S}(0101)$ is the set $\{1a1b:a,b\in\{0,1\}\}$.
If the inputs of the synthetic channels $\{W^{(k)}_{\alpha}\}_{\alpha\in\{0,1\}^{k}}$ are set to be $\{u_{\alpha}\}_{\alpha\in \{0,1\}^{k}}$, then those of the underlying channels $\{W^{(0)}_{\alpha}\}_{\alpha\in \{0,1\}^{k}}$ are given by
\begin{align}
\{x^{(0)}_{\alpha}\}_{\alpha\in \{0,1\}^{k}}=\{u_{\delta}\}_{\delta\in \{0,1\}^{k}}G_{k},
\label{ak05}
\end{align}
where $G_{k}$ is the 0\,-1 matrix of order $2^k$ whose entry at the $\alpha$-column and the $\delta$-row
is 1 if and only if $\delta\in\mathcal{S}(\alpha)$.
For example,
$$G_1=\left(
\begin{array}{cc}
1 & 0 \\
1 & 1 \\
\end{array}
\right),\
G_2=\left(
  \begin{array}{cccc}
    1 & 0 & 0 & 0 \\
    1 & 0 & 1 & 0 \\
    1 & 1 & 0 & 0 \\
    1 & 1 & 1 & 1 \\
  \end{array}
\right),\
G_3=\left(
  \begin{array}{cccccccc}
    1 & 0 & 0 & 0 & 0 & 0 & 0 & 0 \\
    1 & 0 & 0 & 0 & 1 & 0 & 0 & 0 \\
    1 & 0 & 1 & 0 & 0 & 0 & 0 & 0 \\
    1 & 0 & 1 & 0 & 1 & 0 & 1 & 0 \\
    1 & 1 & 0 & 0 & 0 & 0 & 0 & 0 \\
    1 & 1 & 0 & 0 & 1 & 1 & 0 & 0 \\
    1 & 1 & 1 & 1 & 0 & 0 & 0 & 0 \\
    1 & 1 & 1 & 1 & 1 & 1 & 1 & 1 \\
  \end{array}
\right).
$$

The {\it communication scheme} of
the polar code of order $k$ is as follows.
The information sequence $\{u_{\alpha}\}_{\alpha\in \{0,1\}^k}$ is coded into $\{x^{(0)}_{\alpha}\}_{\alpha\in \{0,1\}^k}$ according to (\ref{ak05}) and then transmitted over the BIDMCs $\{W^{(0)}_{\alpha}\}_{\alpha\in \{0,1\}^k}$ respectively, whereas the receiver decodes the synthetic channels $\{W^{(k)}_{\alpha}\}_{\alpha\in \{0,1\}^k}$ sequentially by the {\it Successive Cancellation} (SC) decoding, which generates estimates for $\{u_{\alpha}\}_{\alpha\in \{0,1\}^k}$ by MLD in successive order.
If a \emph{genie} is willing to show us the true inputs of the noisy synthetic channels, which are called {\it frozen channels}, then reliable communication can be realized by decoding only the remaining synthetic channels. Notice that the role of such a genie can be almost realized by a pseudo-random sequence of very large period, which is known at the two ends of communication.

To determine the frozen channels of polar codes, it is desired to compute the capacities of the synthetic channels $\{A_{\alpha }(W)\}_{\alpha\in \{0,1\}^k}$ for large $k$. We will consider to give the compactest RSC forms of these synthetic channels.

\subsection{Arikan Transformations $\Delta_{m}(W)$ and $\nabla_{m}(W)$}

If $\{W_i\}_{i\in[t]}$ are independent copies of some
symmetric BIDMC $W$, we also write $\Delta(\{W_i\}_{i\in[t]})$ and $\nabla(\{W_i\}_{i\in[t]})$
as $\Delta_t(W)$ and $\nabla_t(W)$, respectively. We will deal with these Arikan transformations in this subsection.

According to Lemma~\ref{lem301}, Theorem~\ref{lem200} and (\ref{u01}) to (\ref{u04}) one can deduce the following lemma easily.
\begin{lemma}
Suppose $W$ is a symmetric BIDMC. Let $\Delta_0(W)=\mathrm{B}(0)$ and $\nabla_0(W)=\mathrm{B}(1/2)$.
Then,

\noindent
1.
For nonnegative integers $i,j$, we have
\begin{gather}
A_0(\Delta_i(W),\Delta_j(W))=\Delta_{i+j}(W),\
A_1(\nabla_i(W),\nabla_j(W))=\nabla_{i+j}(W).\label{tt1}
\end{gather}

\noindent
2. For positive integer $k$, we have
\begin{gather}
A_{0^k}(W)=\Delta_{2^k}(W),\
A_{1^k}(W)=\nabla_{2^k}(W).\label{tt11}
\end{gather}
3. For nonnegative numbers $p,q$ with $p+q\leq 1$, let $W_{p,q}$ be
the symmetric BIDMC $p\mathrm{B}(0)+q\mathrm{B}(1/2)+\overline{p+q}W$. For positive integer $t$, we have
\begin{gather}
\Delta_{t}(W_{p,q})\cong \overline{\overline{q}^t}\mathrm{B}\Big(\frac{1}{2}\Big)+
\sum_{i=0}^{t}\binom{t}{i}p^{t-i}\overline{p+q}^i\Delta_{i}(W),\label{pp7}\\
\nabla_{t}(W_{p,q})\cong \overline{\overline{p}^t}\mathrm{B}(0)+
\sum_{i=0}^{t}\binom{t}{i}q^{t-i}\overline{p+q}^i\nabla_{i}(W).\label{pp8}
\end{gather}
\end{lemma}

The following corollary is a direct application of Theorem~\ref{cor501'}.
\begin{corollary}\label{lem501}
For $W\cong\sum_{j\in[n]}q_j\mathrm{B}(\varepsilon_j)\in
\mathbb{B}_n^*$ with $\varepsilon_0^n\in (0,1)^n$ and $m\geq 1$,
\begin{gather}
\Delta_m(W)\cong\sum_{i_0^m\in[n]^m}\mathrm{B}(\star_{j\in[m]}\varepsilon_{i_j})\prod_{j\in[m]}q_{i_j},
\label{mm01}\\
\nabla_m(W)\cong\sum_{\{i_l\in[n],\sigma_{l}\in\{\varepsilon_{i_l},\overline{\varepsilon_{i_l}}\}\}_{
l\in[m]}}
\mathrm{B}(\diamond_{l\in[m]}\sigma_{l})\varpi(\sigma_0^m)\prod_{l\in [m]}q_{i_l},\label{mm02}
\end{gather}
where $\varepsilon$ and $\overline{\varepsilon}$ should be seen as different elements even if $\varepsilon=1/2$.
\end{corollary}

Clearly, according to the properties of operations $\star, \diamond$ one can express $\Delta_m(W), \nabla_m(W)$ as RSCs of less BSCs than those shown in Corollary~\ref{lem501} by merging the equivalent sub-channels.
For $W\in\mathbb{B}_n^*$, let $\phi(W)=n$.

\begin{lemma}\label{lemf00}
1. For any vector $a_0^n$ of nonnegative integers with $\sum_{i\in[n]}a_i=m$, let $K(a_0^n)$ denote the number of vectors in
$$\{d_0^m\in[n]^m:|\{j\in [m]: d_j=i\}|=a_i,\,i\in[n]\}.$$ Then,
\begin{gather} K(a_0^n)=\binom{m}{a_0^n}=\frac{m!}{\prod_{i\in[n]}(a_i!)}.\label{gg20}
\end{gather}
2. Let $H_{m,n}$ denote the number of methods for putting $m$ balls into $n$ distinct boxes. Then,
\begin{align}
H_{m,n}=\binom{m+n-1}{m}=\frac{(m+n-1)!}{m!(n-1)!}.
\end{align}
\end{lemma}
\begin{proof}
The proof for this lemma is simple, we omit it here.
\end{proof}
According to Corollary~\ref{lem501} and Lemma \ref{lemf00}, one can deduce the following theorem easily.
\begin{theorem}\label{cor10}
For $m\geq 1$ and $W\cong\sum_{i\in[n]}q_i\mathrm{B}(\varepsilon_i)\in\mathbb{B}_n^*$ with $\varepsilon_0^n\in (0,1)^n$,

\noindent 1.
$\Delta_m(W)$ is equivalent to an RSC of $\phi(\Delta_m(W))\leq \binom{m+n-1}{m}$ BSCs as the following
\begin{align}
\Delta_m(W)\cong\sum_{\sum_{i\in[n]}a_i=m}
\mathrm{B}\big(\star_{i\in[n]}\varepsilon_{i}^{\star a_i}\big)\binom{m}{a_0^n}\prod_{i\in[n]}q_{i}^{a_i},\label{pp10}
\end{align}
where $a_0^n$ is a vector of nonnegative integers, $\varepsilon^{\star 0}=0$ and, for integer $a>0$,
$$\varepsilon^{\star a}=\underbrace{\varepsilon\star\cdots\star\varepsilon}_{a}
=\sum_{0\leq i\leq \lfloor (a-1)/2\rfloor}\binom{a}{2i+1}\varepsilon^{2i+1}\overline{\varepsilon}^{a-2i-1}.$$

\noindent 2. $\nabla_m(W)$ is equivalent to an RSC of
\begin{gather}
\phi(\nabla_m(W))\leq
1+\sum_{\omega=1}^{m}2^{\omega-1}\binom{n}{\omega}\sum_{0\leq 2b\leq m-\omega}\binom{m-2b-1}{m-2b-\omega}
\end{gather}
BSCs as the following
\begin{align}
\nabla_m(W)\cong&
\sum_{\Omega\subset[n],\sum_{i\in\Omega}a_i+2\sum_{j\in[n]}b_j=m}
\binom{m}{a_0^n+b_0^n,b_0^n}
\prod_{i\in[n]}q_{i}^{a_i+2b_i}(\varepsilon_i\overline{\varepsilon_i})^{b_i}\nonumber\\
&\ \ \ \ \ \ \ \ \ \ \
\sum_{\{\sigma_{l}\in\{\varepsilon_{l},\overline{\varepsilon}_{l}\}\}_{l\in\Omega}}
\frac{1}{2}\bigg(\prod_{i\in\Omega}\sigma_i^{a_i}+
\prod_{i\in\Omega}\overline{\sigma_i}^{a_i}\bigg)\mathrm{B}(\diamond_{i\in\Omega}\sigma^{\diamond a_i}_{i}),
\label{pp4}
\end{align}
where $\varepsilon$ and $\overline{\varepsilon}$ should be always seen as different elements, $b_j\geq 0$ for $j\in[n]$, $a_i>0$ for $i\in\Omega$, $a_l=0$ for $l\in[n]\setminus\Omega$,
$$\sigma^{\diamond a}=\underbrace{\sigma\diamond\cdots\diamond\sigma}_{a}=
\frac{\sigma^a}{\sigma^a+\overline{\sigma}^a},$$
and the inner summation in (\ref{pp4}) denotes $\mathrm{B}(1/2)$ if $\Omega=\emptyset$.
\end{theorem}
\begin{proof}
(\ref{pp10}) and (\ref{pp4}) follow
simply from Corollary~\ref{lem501} and Lemma \ref{lemf00}.
According to (\ref{pp10}) we see $\phi(\Delta_m(W))\leq H_{m,n}=\binom{m+n-1}{m}$.
According to (\ref{pp4}) and $\mathrm{B}(\diamond_{i\in\Omega}\sigma^{\diamond a_i}_{i})\cong\mathrm{B}(\diamond_{i\in\Omega}\overline{\sigma}^{\diamond a_i}_{i})$ we see
\begin{align*}
\phi(\nabla_m(W))\leq& 1+\sum_{\omega=1}^{m}2^{\omega-1}\binom{n}{\omega}\sum_{0\leq 2b\leq m-\omega}H_{m-\omega-2b,\omega}\\
=&1+\sum_{\omega=1}^{m}2^{\omega-1}\binom{n}{\omega}\sum_{0\leq 2b\leq m-\omega}\binom{m-2b-1}{m-2b-\omega},
\end{align*}
where $\omega$ and $b$ correspond the cardinality of the set $\Omega$ and the summation $\sum_{i\in[n]}b_i$, respectively.
\end{proof}

If some of the cross-over probabilities are equal to $1/2$, then Theorem \ref{cor10}
can be refined further as the following theorem.
\begin{theorem}
\label{lem80}
For $m\geq 1$ and $W\cong\sum_{j\in[n+1]}q_j\mathrm{B}(\varepsilon_j)\in\mathbb{B}_{n+1}^*$ with $0<\varepsilon_0<\cdots<\varepsilon_{n-1}<\varepsilon_n=1/2$,

\noindent 1.
$\Delta_{m}(W)$ is equivalent to an RSC of $\phi(\Delta_{m}(W))\leq H_{m,n-1}+1=\binom{m +n-1}{m}+1$ BSCs as the following
\begin{align}
\Delta_{m}(W)\cong&
(1-\overline{q_n}^{m})\mathrm{B}\big(\frac{1}{2}\big)+\sum_{\sum_{i\in[n]}a_i=m}
\mathrm{B}\big(\star_{i\in[n]}\varepsilon_{i}^{\star a_i}\big)
\binom{m}{a_0^n}\prod_{i\in[n]}q_{i}^{a_i},
\label{pp11}
\end{align}
where $a_0^n$ is vector of nonnegative integers.

\noindent 2. $\nabla_{m}(W)$ is equivalent to an RSC of
\begin{gather}
\phi(\nabla_{m}(W))\leq 1+\sum_{\omega=1}^{m}2^{\omega-1}\binom{n}{\omega}\binom{m}{\omega}
\end{gather}
BSCs as the following
\begin{align}
\nabla_{m}(W)\cong&
\sum_{s+\sum_{j\in[n]}(a_j+2b_j)=m}
\binom{m}{s,a_0^n+b_0^n,b_0^n}
q_n^s\prod_{i\in[n]}q_{i}^{a_i+2b_i}(\varepsilon_i\overline{\varepsilon_i})^{b_i}\nonumber\\
&\ \ \ \ \ \
\sum_{\{\sigma_l\in\{\varepsilon_l,\overline{\varepsilon}_l\}\}_{a_l>0,l\in[n]}}
\frac{1}{2}\bigg(\prod_{i\in[n]}\sigma_i^{a_i}+\prod_{i\in[n]}\overline{\sigma_i}^{a_i}\bigg)
\mathrm{B}(\diamond_{i\in[n]}\sigma^{\diamond a_i}_{i}),\label{pp5}
\end{align}
where integer $s\geq 0$, $a_0^n,b_0^n$ are vectors of nonnegative integers, $\sigma^{\diamond 0}=1/2$ and the inner summation in (\ref{pp5}) denotes $\mathrm{B}(1/2)$ if $\sum_{i\in[n]}a_i=0$.
\end{theorem}
\begin{proof}
Clearly, (\ref{pp11}) follows from (\ref{pp7}) and (\ref{pp10}).
For any vectors $a_0^n, b_0^n$ of nonnegative integers with $\sum_{i\in[n]}(a_i+2b_i)\leq m$,
let $s=m-\sum_{i\in[n]}(a_i+2b_i)$. If $s=2t+1$ is odd, then
we have
\begin{align*}
&2\sum_{b_n=0}^t\binom{m}{a_0^n+b_0^n,b_0^n}=
2^{s}\binom{m}{s,a_0^n+b_0^n,b_0^n},
\end{align*}
where $a_n=s-2b_n$ is always positive. If $s=2t$ is even, then we have
\begin{align*}
&\binom{m}{a_0^n+b_0^n,b_0^n,t,t}+
2\sum_{b_n=0}^{t-1}\binom{m}{a_0^{n+1}+b_0^{n+1},b_0^{n+1}}
=2^{s}\binom{m}{s,a_0^n+b_0^n,b_0^n},
\end{align*}
where $a_n=s-2b_n$ is always positive too. Therefore, we see that (\ref{pp5}) follows from (\ref{pp4}), $\varepsilon_n=\frac{1}{2}$ and $\sigma\diamond \frac{1}{2}=\sigma$ for any $\sigma\in(0,1)$.
Furthermore, from (\ref{pp5}) and $\overline{\diamond_{i\in[n]}\sigma^{\diamond a_i}_{i}}=\diamond_{i\in[n]}\overline{\sigma}^{\diamond a_i}_{i}$ we see
\begin{align*}
\phi(\nabla_m(W))\leq& 1+\sum_{\omega=1}^{m}2^{\omega-1}\binom{n}{\omega}\sum_{0\leq k\leq m-\omega}H_{m-\omega-k,\omega}\\
=&1+\sum_{\omega=1}^{m}2^{\omega-1}\binom{n}{\omega}\sum_{0\leq k\leq m-\omega}\binom{m-k-1}{m-k-\omega}\\
=&1+\sum_{\omega=1}^{m}2^{\omega-1}\binom{n}{\omega}\binom{m}{\omega},
\end{align*}
where $\omega$ corresponds the number of positive integers $a_i$, $k$ corresponds the sum of $s$ and $2\sum_{i\in[n]}b_i$.
\end{proof}

We note that, for any synthetic channel generated in polar codes over symmetric BIDMCs,
it is not difficult to give an exact expression of RSC of BSCs for it by using
Theorems~\ref{cor10} and \ref{lem80}.

\subsection{Arikan Transformations $A_{\alpha}(\mathrm{E}(q))$ and $A_{\alpha}(\mathrm{B}(\varepsilon))$}

It has been pointed out in \cite{Arikan09} that all the Arikan transformations are BECs when
the underlying channel is a BEC.
Indeed, if the underlying channel $W$ is $\mathrm{E}(q)\cong\overline{q}\mathrm{B}(0)+q\mathrm{B}(1/2)$, the BEC with erasure probability $q$, then for any integer $t$ from (\ref{pp7}) and (\ref{pp8}) we have
\begin{gather}
\Delta_{t}(\mathrm{E}(q))\cong
\overline{q}^t\mathrm{B}(0)+\overline{\overline{q}^t}\mathrm{B}(1/2)
\cong\mathrm{E}(\overline{\overline{q}^t}), \\
\nabla_{t}(\mathrm{E}(q))\cong
\overline{q^t}\mathrm{B}(0)+q^t\mathrm{B}(1/2)\cong\mathrm{E}(q^t).
\end{gather}
Furthermore, for nonnegative integer $s$, let $f_s$ denote the map over $[0,1]$ defined by $f_s(p)=\overline{p}^{2^s}$.
Therefore, for any sequence $\alpha=0^{t_1}1^{t_2}\cdots0^{t_{2r-1}}1^{t_{2r}}$ in $\{0,1\}^{*}$ with $t_1\geq 0,t_{2r}\geq 0$ and $t_i>0, 2\leq i\leq 2r-1$, we have
\begin{gather}
A_{\alpha}(\mathrm{E}(q))\cong \overline{F_{\alpha}(q)}\mathrm{B}(0)
+F_{\alpha}(q)\mathrm{B}(1/2)
\cong \mathrm{E}(F_{\alpha}(q)),\\
I(A_{\alpha}(\mathrm{E}(q)))=\overline{F_{\alpha}(q)},
\end{gather}
where $F_{\alpha}$ is the compound map $f_{t_{2r}}\circ f_{t_{2r-1}}\circ\cdots \circ f_{t_1}$ over $[0,1]$.
For example, for $\alpha=0110=0^11^20^11^0$, we have
\begin{gather*}
F_{0110}(q)=1-(1-(1-(1-q)^2)^4)^2,\\
I(A_{0110}(\mathrm{E}(q)))=\overline{F_{0110}(q)}=(1-(1-(1-q)^2)^4)^2.
\end{gather*}

If the underlying channel $W$ is a general BIDMC but BEC,
for any Arikan transformation of $W$ one can deduce an exact expression of RSC of BSCs by using
the results shown in the last subsection.
However, as pointed in the proof of Theorem~\ref{lem80}, the Arikan transformation $\nabla_{m}(W)$ given in (\ref{pp5}) is not in the compactest form, some of the sub-channels may be merged further.

In the following theorem we show the compactest forms for some of the synthetic channels generated in polar codes when the underlying channel is a BSC.

\begin{theorem}
For $l\geq 0$ and $\varepsilon\in(0,1/2)$, we have
\begin{gather}
A_{0^l}(\mathrm{B}(\varepsilon))\cong \mathrm{B}(\varepsilon_l),\label{mm307'}
\end{gather}
where $\varepsilon_l=\varepsilon^{\star 2^l}$. For $k\geq 0$, we have
\begin{align}
A_{0^l1^{k+1}}(\mathrm{B}(\varepsilon))\cong
\binom{2^{k+1}}{2^{k}}(\varepsilon_l\overline{\varepsilon_l})^{2^{k}}
\mathrm{B}\Big(\frac{1}{2}\Big)+
\sum_{i=1}^{2^{k}}\binom{2^{k+1}}{2^{k}-i}
\frac{\varepsilon_l^{2i}+\overline{\varepsilon_l}^{2i}}
{(\varepsilon_l\overline{\varepsilon_l})^{i-2^k}}
\mathrm{B}(\varepsilon_l^{\diamond (2i)}).\label{pp0}
\end{align}
For $i\geq 0$, let $\varepsilon_{l,i}=(\varepsilon_l^{\diamond 2})^{\star 2^i}$ and $q_{l,i}
=(\varepsilon_l^{2}+\overline{\varepsilon_l}^{2})^{2^i}$. Then, for $k\geq 0$, we have
\begin{align}
A_{0^l10^i1^k}(\mathrm{B}(\varepsilon))
\cong&
\sum_{0\leq b\leq 2^{k-1}}
\binom{2^k}{2^k-2b,b,b}
\overline{q_{l,i}}^{2^k-2b}(q_{l,i}^2\varepsilon_{l,i}\overline{\varepsilon_{l,i}})^{b}
\mathrm{B}\Big(\frac{1}{2}\Big)\nonumber\\
+\sum_{a=1}^{2^{k}}\mathrm{B}(\varepsilon_{l,i}^{\diamond a})&
(\varepsilon_{l,i}^{a}+\overline{\varepsilon_{l,i}}^{a})q_{l,i}^a
\sum_{s+2b=2^k-a}
\binom{2^k}{s,a+b,b}
\overline{q_{l,i}}^s(q_{l,i}^2\varepsilon_{l,i}\overline{\varepsilon_{l,i}})^{b}.\label{pp1'''}
\end{align}
Let
$p_{l,i}=q_{l,i}^2(\varepsilon_{l,i}^2+\overline{\varepsilon_{l,i}}^2)+2q_{l,i}\overline{q_{l,i}}$ and
$r_{l,i}=2q_{l,i}\overline{q_{l,i}}/p_{l,i}$. Then, for $t\geq 0$, we have
\begin{align}
&A_{0^l10^i10^t}(\mathrm{B}(\varepsilon))\cong
a_{l,i,t}\mathrm{B}\Big(\frac{1}{2}\Big)+
\sum_{s=0}^{2^t}b_{l,i,t,s}
\mathrm{B}\big(\beta_{l,i,t,s}\big),\label{pp2''}
\end{align}
where $a_{l,i,t}=1-p_{l,i}^{2^t}$,
$b_{l,i,t,s}=\binom{2^t}{s}\overline{r_{l,i}}^{2^t-s}r_{l,i}^{s}p_{l,i}^{2^t}$ and
$\beta_{l,i,t,s}=(\varepsilon_{l,i}^{\diamond 2})^{\star(2^t-s)}\star(\varepsilon_{l,i})^{\star s}$ for $0\leq s\leq 2^t$. Furthermore, we have $0<\beta_{l,i,t,0}<\beta_{l,i,t,1}<\cdots<\beta_{l,i,t,2^t}<1/2$ and
\begin{align}
&A_{0^l10^i10^t1}(\mathrm{B}(\varepsilon))
\cong
\Big(a_{l,i,t}^2+2\sum_{s=0}^{2^t}b_{l,i,t,s}^2\beta_{l,i,t,s}\overline{\beta_{l,i,t,s}}\Big)
\mathrm{B}\Big(\frac{1}{2}\Big)+\nonumber\\
&\ \ \ \ 2a_{l,i,t}\sum_{s=0}^{2^t}b_{l,i,t,s}\mathrm{B}(\beta_{l,i,t,s})
+\sum_{s=0}^{2^t}b_{l,i,t,s}^2(\beta_{l,i,t,s}^2+\overline{\beta_{l,i,t,s}}^2)
\mathrm{B}(\beta_{l,i,t,s}^{\diamond 2})+
\nonumber\\
&\ \ \ \ 2\sum_{0\leq s<r\leq 2^t}b_{l,i,t,s}b_{l,i,t,r}\sum_{\sigma\in\{\beta_{l,i,t,r},\overline{\beta_{l,i,t,r}}\}}
(\beta_{l,i,t,s}\sigma+\overline{\beta_{l,i,t,s}}\,\overline{\sigma})
\mathrm{B}(\beta_{l,i,t,s}\diamond\sigma).
\label{h001}
\end{align}
\end{theorem}
\begin{proof}
From (\ref{tt11}) and (\ref{pp10}), we see (\ref{mm307'}).

From (\ref{tt11}), (\ref{pp4}) and (\ref{mm307'}), we see
\begin{align*}
&A_{0^l1^{k+1}}(\mathrm{B}(\varepsilon)) \cong A_{1^{k+1}}(A_{0^l}(\mathrm{B}(\varepsilon))) \cong A_{1^{k+1}}(\mathrm{B}(\varepsilon_l))\cong\nabla_{2^{k+1}}(\mathrm{B}(\varepsilon_l))\\
\cong&
\sum_{a+2b=2^{k+1}}\binom{2^{k+1}}{a+b,b}(\varepsilon_l\overline{\varepsilon_l})^b
\sum_{\sigma\in\{\varepsilon_l,\overline{\varepsilon_l}\}}\frac{\sigma^a+\overline{\sigma}^a}{2}
\mathrm{B}(\sigma^{\diamond a}),
\end{align*}
where the inner summation denotes $\mathrm{B}(1/2)$ when $a=0$.
Therefore, by taking $a=2i$ and $b=2^k-i$ we see that (\ref{pp0}) follows from $\mathrm{B}(\sigma^{\diamond a})\cong\mathrm{B}(\overline{\sigma}^{\diamond a})$.

From (\ref{pp7}), (\ref{mm307'}) and (\ref{pp0}) we have
\begin{align*}
&A_{0^l10^i1^k}(\mathrm{B}(\varepsilon))
\cong A_{0^i1^k}(A_{0^l1}(\mathrm{B}(\varepsilon)))\cong
A_{0^i1^k}(2\varepsilon_l\overline{\varepsilon_l}\mathrm{B}(1/2)
+(\varepsilon_l^2+\overline{\varepsilon_l}^2)\mathrm{B}(\varepsilon_l^{\diamond 2}))\nonumber\\
\cong &
A_{1^k}(\overline{q_{l,i}}\mathrm{B}(1/2)+q_{l,i}\mathrm{B}(\varepsilon_{l,i}))
\cong \nabla_{2^k}(\overline{q_{l,i}}\mathrm{B}(1/2)+q_{l,i}\mathrm{B}(\varepsilon_{l,i}))
\nonumber\\
\cong&
\sum_{0\leq b\leq 2^{k-1}}
\binom{2^k}{2^k-2b,b,b}
\overline{q_{l,i}}^{2^k-2b}q_{l,i}^{2b}(\varepsilon_{l,i}\overline{\varepsilon_{l,i}})^{b}
\mathrm{B}\Big(\frac{1}{2}\Big)\nonumber\\
&+\sum_{a=1}^{2^{k}}\sum_{\sigma\in\{\varepsilon_{l,i},\overline{\varepsilon_{l,i}}\}}
\frac{1}{2}\bigg(\sigma^{a}+\overline{\sigma}^{a}\bigg)
\mathrm{B}(\sigma^{\diamond a})\sum_{s+2b=2^k-a}
\binom{2^k}{s,a+b,b}
\overline{q_{l,i}}^sq_{l,i}^{a+2b}(\varepsilon_{l,i}\overline{\varepsilon_{l,i}})^{b}
\nonumber\\
\cong&
\sum_{0\leq b\leq 2^{k-1}}
\binom{2^k}{2^k-2b,b,b}
\overline{q_{l,i}}^{2^k-2b}(q_{l,i}^2\varepsilon_{l,i}\overline{\varepsilon_{l,i}})^{b}
\mathrm{B}\Big(\frac{1}{2}\Big)\nonumber\\
&+\sum_{a=1}^{2^{k}}\mathrm{B}(\varepsilon_{l,i}^{\diamond a})
(\varepsilon_{l,i}^{a}+\overline{\varepsilon_{l,i}}^{a})q_{l,i}^a
\sum_{s+2b=2^k-a}
\binom{2^k}{s,a+b,b}
\overline{q_{l,i}}^s(q_{l,i}^2\varepsilon_{l,i}\overline{\varepsilon_{l,i}})^{b},
\end{align*}
i.e. (\ref{pp1'''}) is valid for $k\geq 0$.

From (\ref{pp7}), (\ref{pp10}) and (\ref{pp1'''}) we have
\begin{align*}
&A_{0^l10^i10^t}(\mathrm{B}(\varepsilon))=A_{0^t}(A_{0^l10^i1}(\mathrm{B}(\varepsilon)))\nonumber\\
\cong& A_{0^t}\Big(q_{l,i}^2(\varepsilon_{l,i}^2+\overline{\varepsilon_{l,i}}^2)
\mathrm{B}(\varepsilon_{l,i}^{\diamond 2})
+2q_{l,i}\overline{q_{l,i}}\mathrm{B}(\varepsilon_{l,i})+
(2\varepsilon_{l,i}\overline{\varepsilon_{l,i}}q_{l,i}^2+\overline{q_{l,i}}^2)
\mathrm{B}(1/2)\Big)\nonumber\\
\cong&\Delta_{2^t}\Big(p_{l,i}\big(\overline{r_{l,i}}\mathrm{B}(\varepsilon_{l,i}^{\diamond 2})
+r_{l,i}\mathrm{B}(\varepsilon_{l,i})\big)+\overline{p_{l,i}}\mathrm{B}(1/2)\Big)\nonumber\\
\cong&(1-p_{l,i}^{2^t})\mathrm{B}(1/2)+p_{l,i}^{2^t}\Delta_{2^t}
\big(\overline{r_{l,i}}\mathrm{B}(\varepsilon_{l,i}^{\diamond 2})
+r_{l,i}\mathrm{B}(\varepsilon_{l,i})\big)\nonumber\\
\cong &
\big(1-p_{l,i}^{2^t}\big)\mathrm{B}\Big(\frac{1}{2}\Big)+p_{l,i}^{2^t}
\sum_{s=0}^{2^t}\binom{2^t}{s}\overline{r_{l,i}}^{2^t-s}r_{l,i}^{s}
\mathrm{B}\big((\varepsilon_{l,i}^{\diamond 2})^{\star (2^t-s)}\star\varepsilon_{l,i}^{\star s}\big),
\end{align*}
i.e. (\ref{pp2''}) is valid for any $t\geq 0$.

Since for any $\alpha,\sigma\in(0,1/2)$ we have $$1/2>\alpha\star\sigma>\max\{\alpha,\sigma\}\geq\min\{\alpha,\sigma\}>\alpha\diamond\sigma>0,$$ we see easily $0<\beta_{l,i,t,0}<\beta_{l,i,t,1}<\cdots<\beta_{l,i,t,2^t}<1/2$.
Therefore, (\ref{h001}) follows from (\ref{at01}), (\ref{at32}) and (\ref{pp2''}).
\end{proof}

Notice that, by taking $i=0$ in (\ref{pp1'''}), from (\ref{pp0}) one can deduce easily that,
for $k\geq 0$ and $0\leq a\leq 2^k$,
\begin{gather}
\sum_{s+2b=2^k-a}\binom{2^k}{s,a+b,b}2^{s}=\binom{2^{k+1}}{2^k-a},
\end{gather}
where $s, b$ assume nonnegative integers.

\subsection{Number of BSCs in Arikan Transformations}
Let $0\leq\varepsilon_0<\cdots<\varepsilon_{n-1}<\varepsilon_n=\frac{1}{2}$ and
$W\cong\sum_{j\in[n+1]}q_j\mathrm{B}(\varepsilon_j)$.
Clearly, the sizes of the output sets of $W$, $A_0(W)$, $A_1(W)$, $A_{00}(W)$, $A_{01}(W)$, $A_{10}(W)$, $A_{11}(W)$ are at least $2n+1$, $(2n+1)^2$, $2(2n+1)^2$, $(2n+1)^4$, $2(2n+1)^4$, $4(2n+1)^4$, $8(2n+1)^4$, respectively.
However, on the numbers of RSCs in their LRP-oriented forms, according to Theorem~\ref{lem80}, we see
\begin{align}\label{pp60}
\phi(A_0(W))&\leq \frac{n^2+n}{2}+1,\\
\phi(A_1(W))&\leq n^2+n+1,\label{pp61}\\
\phi(A_{00}(W))&\leq \binom{n+3}{4}+1=\frac{(n^2+n)(n^2+5n+6)}{24}+1,\label{pp62}\\
\phi(A_{11}(W))&\leq 4\binom{n}{1}+12\binom{n}{2}+16\binom{n}{3}+8\binom{n}{4}+1\nonumber\\
&=\frac{(n^2+n)(n^2+n+4)}{3}+1,\label{pp63}
\end{align}
and by using (\ref{pp60}) and (\ref{pp61}) we have
\begin{align}\label{pp66}
\phi(A_{01}(W))&\leq \frac{(n^2+n)(n^2+n+2)}{4}+1,\\
\phi(A_{10}(W))&\leq \frac{(n^2+n)(n^2+n+1)}{2}+1.
\end{align}

Moreover, for $\alpha\in \{0,1\}^k$, let $\varphi(\alpha)$ denote the number defined by $\varphi(1^l)=\varphi(0^l)=(2^l)!2^{2^l}$ for $l\geq 0$ and, for $\beta\in \{0,1\}^s$,
\begin{gather*}
\varphi(\beta 01^l)=(2^l)!(\varphi(\beta 0))^{2^l},\
\varphi(\beta 10^l)=(2^l)!(\varphi(\beta 1))^{2^l}.
\end{gather*}
Then, we have
\begin{theorem}
Let $0\leq\varepsilon_0<\cdots<\varepsilon_{n-1}<\varepsilon_n=\frac{1}{2}$ and
$W\cong\sum_{j\in[n+1]}q_j\mathrm{B}(\varepsilon_j)$.
Then, for $\alpha\in \{0,1\}^k$, we have
\begin{align}
\label{f001}
\phi(A_{\alpha}(W))\leq h_{\alpha}(n)=\frac{2^{b(\alpha)}(2n)^{2^k}}{\varphi(\alpha)}+g_{\alpha}(n),
\end{align}
where $g_{\alpha}(n)$ is a polynomial of order at most $2^k-1$ and
$b(\alpha)=a_{k-1}2^{k-1}+a_{k-2}2^{k-2}+\cdots+a_{0}$ if $\alpha=a_{k-1}a_{k-2}\cdots a_0$.
In particular, for any $\alpha\in \{0,1\}^k$ the average number of elements with the same likelihood ratios in the output set of $A_{\alpha}(W)$
is at least $\varphi(\alpha)/2$ when $n$ is sufficiently large.
\end{theorem}
\begin{proof}
If $\alpha=0^k$, then we have $b(\alpha)=0$, $\varphi(\alpha)=(2^k)!2^{2^k}$ and, according to (\ref{tt11}) and Theorem~\ref{lem80},
there is a polynomial $g_{0^k}(n)$ of order at most $2^k-1$ such that
\begin{align*}
\phi(A_{\alpha}(W))\leq 1+\binom{n+2^k-1}{2^k}=\frac{n^{2^k}}{(2^k)!}+g_{0^k}(n)
=\frac{2^{b(\alpha)}(2n)^{2^k}}{\varphi(\alpha)}+g_{0^k}(n),
\end{align*}
i.e., (\ref{f001}) is valid for $\alpha=0^k$.

If $\alpha=1^k$, then we have $b(\alpha)=2^k-1$, $\varphi(\alpha)=(2^k)!2^{2^k}$ and, according to (\ref{tt11}) and Theorem~\ref{lem80},
there is a polynomial $g_{1^k}(n)$ of order at most $2^k-1$ such that
\begin{align*}
\phi(A_{\alpha}(W))&\leq
1+\sum_{\omega=1}^{2^k}2^{\omega-1}\binom{n}{\omega}\binom{2^k}{\omega}
=\frac{2^{2^k-1}n^{2^k}}{(2^k)!}+g_{1^k}(n)
=\frac{2^{b(\alpha)}(2n)^{2^k}}{\varphi(\alpha)}+g_{1^k}(n),
\end{align*}
i.e., (\ref{f001}) is valid for $\alpha=1^k$.

Now we consider to prove (\ref{f001}) by induction on $k$.
Assume that (\ref{f001}) is valid for any sequence $\alpha$ of length smaller than $k$.
Suppose $\alpha\in \{0,1\}^k\setminus\{0^k,1^k\}$ and $l$ is the largest integer
such that $\alpha=\sigma\delta$, $\delta\in\{0^l,1^l\}$ and $\sigma\in\{0,1\}^{k-l}$. Then, we have $1\leq l<k$,
$\varphi(\delta)=(2^l)!2^{2^l}$, $\varphi(\alpha)=(2^l)!(\varphi(\sigma))^{2^l}$, $b(\alpha)=b(\sigma)2^l+b(\delta)$ and
\begin{gather*}
\phi(A_{\sigma}(W))\leq h_{\sigma}(n)=\frac{2^{b(\sigma)}(2n)^{2^{k-l}}}{\varphi(\sigma)}+g_{\sigma}(n).
\end{gather*}
Hence, we see there is a polynomial $g_{\alpha}(n)$ of order at most $2^k-1$ such that
\begin{align*}
\phi(A_{\alpha}(W))&=\phi(A_{\delta}(A_{\sigma}(W)))\\
&\leq \frac{2^{b(\delta)}(2h_{\sigma}(n))^{2^l}}{\varphi(\delta)}+g_{\delta}(h_{\sigma}(n))\\
&=\frac{2^{b(\delta)}2^{2^l}2^{b(\sigma)2^l}(2n)^{2^{k-l}2^l}}{(2^l)!2^{2^l}(\varphi(\sigma))^{2^l}}+g_{\alpha}(n)\\
&=\frac{2^{b(\alpha)}(2n)^{2^k}}{\varphi(\alpha)}+g_{\alpha}(n),
\end{align*}
where $g_{\delta}(h_{\sigma}(n))$ is a polynomial of order at most
$(2^l-1)2^{k-l}\leq 2^k-1$. Thus, (\ref{f001}) is valid for any sequence $\alpha$ of length $k$.

Since the size of the output set of $A_{\alpha}(W)$ is $2^{b(\alpha)}(2n+1)^{2^k}$, from (\ref{f001}) we see that the average size of the non-empty sets $L_{A_{\alpha}(W)}(\varepsilon)$
is at least $\varphi(\alpha)/2$ when $n$ is sufficiently large.
\end{proof}

Notice that one can deduce easily that, for each $\alpha\in \{0,1\}^k$, the number $\varphi(\alpha)2^{1-2^{k+1}}$ is an odd integer and
\begin{align*}
2^{2^{k+1}-1}\leq\varphi(\alpha)\leq (2^k)!2^{2^k}.
\end{align*}

\section{Conclusions}
\label{sec05}
The focus of this paper is to investigate the principal properties of the synthetic channels that are iteratively formed through Arikan transformations during the construction of polar codes. Given that evaluating the reliability of these synthetic channels is of great significance.

For binary input discrete memoryless channels (BIDMCs), we defined their equivalence and symmetry based on their likelihood ratio profiles (LRPs). Then, closed-form formulas
for the computation of the LRPs of Arikan transformations of BIDMCs were derived. By utilizing the random switching channels (RSCs) of BIDMCs and converting channel transformations into algebraic operations, we obtained concise and compact expressions for a few Arikan transformations of symmetric BIDMCs. For the case when the underlying channel is a binary symmetric channel, we shown the compactest representations of RSCs for many synthetic channels generated in polar code construction.
In the end, we derived a lower bound for the average number of elements that possess the same likelihood ratio within the output alphabet of any synthetic channel generated in polar code construction.

\end{document}